\documentclass[aos]{imsart}

\RequirePackage{amsthm,amsmath,amsfonts,amssymb}
\RequirePackage[authoryear]{natbib}
\RequirePackage[colorlinks,citecolor=blue,urlcolor=blue]{hyperref}
\RequirePackage{graphicx}
\RequirePackage{subfig}

\startlocaldefs
\theoremstyle{plain}

\newtheorem{theorem}{Theorem}[section]
\newtheorem{lemma}[theorem]{Lemma}
\newtheorem{assumption}{Assumption}
\newtheorem{corollary}{Corollary}
\theoremstyle{remark}

\newtheorem{remark}{Remark}

\DeclareMathOperator{\tr}{\text{tr}}
\usepackage{graphicx}

\newcommand{\mle}{\hat\psi}
\newcommand{\Blbrace}{\Big\lbrace}
\newcommand{\Brbrace}{\Big\rbrace}
\newcommand{\Blp}{\Big(}
\newcommand{\Brp}{\Big)}
\newcommand{\sign}{\text{sign}}
\newcommand{\info}{j_{\lambda\lambda}}
\def\underbrace#1{\mathop{...}\limits}

\newcommand{\norm}[1]{\left\lVert#1\right\rVert}

\endlocaldefs

\begin{document}

\begin{frontmatter}
\title{Asymptotic Behaviour of the Modified Likelihood Root}
\runtitle{Asymptotic Behaviour of the Modified Likelihood Root}

\begin{aug}
\author[A, B]{\fnms{Yanbo} \snm{Tang}},
\and
\author[A]{\fnms{Nancy} \snm{Reid}}
\address[A]{Department of Statistical Sciences,
University of Toronto, Toronto, Canada 
 }
\address[B]{Vector Institute, Toronto, Canada
}
\end{aug}

\begin{abstract}
We examine the normal approximation of the modified likelihood root, an inferential tool from higher-order asymptotic theory, for the linear exponential and location-scale family. We show that the $r^\star$ statistic can be thought of as a location and scale adjustment to the likelihood root $r$ up to $O_p(n^{-3/2})$, and more generally $r^\star$ can be expressed as a polynomial in $r$. We also show the linearity of the modified likelihood root in the likelihood root for these two families. 
\end{abstract}

\end{frontmatter}

\section{ Introduction}

The use of $p$-values, although sometimes controversial, has become a key part of modern statistical science, they are the building block of various multiple testing correction procedures used in statistical genetics, where millions of hypotheses are simultaneously considered. 
In most circumstances $p$-values are not exact but are calculated from the limiting distribution of a test statistic. 
The usual test statistics provided in statistical software, such as the likelihood ratio test, Wald test and score test, all have a common known limiting distribution and are accurate to the first order, meaning that the approximation error behaves like $O(n^{-1/2})$. 
However, in the small sample setting or when the number of nuisance parameters is high relative to the number of observations, this trio of tests may perform poorly. 
An improved test statistic, $r^\star$, a modified version of the likelihood root, can be used for likelihood-based inference for scalar parameters of interest. It produces more accurate $p$-values compared to the first order approximations of the test statistics. 
The accuracy of the $p$-values generated by $r^\star$ can be quite remarkable; in some models the approximation is accurate with very few observations \citep[~\S3.2]{2007applied}.

Given the importance that $p$-values play in statistical inference, the exact mechanism through which $r^\star$ generates more accurate $p$-values warrants a careful examination.
We provide insight into the behaviour of $r^\star$ by expressing it as a formal asymptotic expansion, showing that it is asymptotically linear the likelihood root, which we introduce below.

We assume the data $y = (y_1, \cdots, y_n)^\top$ are generated independently from a model parametrized by $ \theta = (\psi, \lambda)$ where $\psi$ is a scalar parameter of interest, $\lambda$ is the nuisance parameter and $\hat\lambda_{\psi}$ denotes the constrained maximum likelihood estimate; i.e. the value of $\lambda$ that maximises the log-likelihood function for a fixed $\psi$.
We denote the log-likelihood function by $l(\psi,\lambda; y)$, and the
the data generating parameter by $\theta_0 = (\psi_0, \lambda_0)$.
The profile log-likelihood function,
\[l_\text{p}(\psi; y) :=  \sup_{\lambda} l(\psi, \lambda; y) = l(\psi, 
\hat\lambda_{\psi}; y) , \]
accounts for the presence of nuisance parameters through constrained maximization.  Under suitable regularity conditions \cite[~\S 3.4]{barndorff1994asymptotic},
\[ w(\psi_0; y) := 2 \lbrace l_\text{p}(\hat\psi) - l_\text{p}(\psi_0) \rbrace \xrightarrow{d} \chi^2_1, \]
where $\chi^2_1$ is a random variable distributed as chi-squared with one degree of freedom.
Equivalently, the log-likelihood root
\begin{align}
    r(\psi_0; y) &:= \sign(\hat\psi - \psi_0)[2\lbrace l_\text{p}(\hat\psi; y) - l_\text{p}(\psi_0; y)  \rbrace]^{\frac{1}{2}} \xrightarrow{d} Z, \label{eq:r}
\end{align}
where the random variable $Z$ has the standard normal distribution \cite[~\S 2.3]{barndorff1994asymptotic}. 

By adding a correction term to $r$, we obtain the modified likelihood root 
\begin{align}\label{eq:intro+rstar}
 r^\star(\psi_0; y) = r(\psi_0; y)  + \frac{1}{r(\psi_0; y)}\log\left\lbrace \frac{q(\psi_0; y)}{r(\psi_0; y)} \right\rbrace;
\end{align}
the form of $q$ depends on the model. 
It has been shown under regularity conditions that the normal approximation to the distribution of $r^\star$ is accurate to $O(n^{-3/2})$ \citep{barndorff1994asymptotic}, whereas the normal approximation to $r$ is only accurate to $O(n^{-1/2})$.   

The mechanism through which $r^\star$ achieves this accuracy is not entirely transparent.
\cite{CAKMAK1998211} show that in models with no nuisance parameters, the leading term in the adjustment factor $r^\star$ is a location and scale correction to $r$. 
We generalize this result to models with nuisance parameters. 
We then use a formal expansion to discuss the asymptotic behaviour of $r^\star$ when the number of parameters increases with the number of observations. 
We show that the adjustment factor $(1/r)\log(q/r)$ is potentially of the same asymptotic order as $r$ in the high-dimensional regime, agreeing with results in \cite{Tang_Reid}. 

Though similar techniques as in the proof of the main theorems. We also show that $r = q  + q^2A/n^{1/2} + q^3B/n + O_p(n^{-3/2}) $, in the linear exponential and location scale families, where $A$ and $B$ are $O_p(1)$, which may depend on $y$ and $\theta_0$. 
This result underlies the proof that the normal approximation to the distribution of $r^\star$ has relative error $O(n^{-3/2})$ \citep[~\S8.5]{2007applied}, but to our knowledge this has not been established in the vector parameter setting.

We focus our analysis on the location-scale and linear exponential family, as the expression for $q$ in (\ref{eq:intro+rstar}) are explicitly available.
We begin with background details on $r^\star$, the linear exponential family and the location-scale family in \S \ref{Sec2}.
We present out main theorems in \S \ref{Sec3}, briefly discuss the impact of the results if $p$ scales with $n$ in \S\ref{Sec4} and show that $r$ can be expressed as a third-order polynomial in $q$ in \S\ref{Sec5}.   
We present some simulations in \S\ref{Sec6} which illustrate the accuracy of the approximations to $r^\star$ and conclude with some additional proof details in \S\ref{Appendix}.  

\subsection{Notation}
Derivatives of the log-likelihood function are denoted by subscripts placed on ``$l$", for example $l_{\psi\lambda\lambda}(\theta)$ represents the matrix with components are $ [l_{\psi\lambda\lambda}(\theta)]_{rs} =  \partial^3 l(\theta)/\partial\psi\partial\lambda_r\partial\lambda_s$. 
We use $j$ to denote the observed information matrix, $j(\theta) = - l_{\theta\theta}(\theta)$; subscripts placed on $j$ denote sub-matrices of $j$ and we let $j_p(\psi) = -d^2 l_{\text{p}}(\psi)/d\psi^2$.
A tilde placed on any quantity denotes that it is evaluated at the constrained maximum likelihood estimate, $(\psi, \hat\lambda_\psi)$, for example $\tilde{\jmath}_{\lambda\lambda} = j_{\lambda\lambda}(\psi, \hat\lambda_\psi)$,
and a hat denotes that it is evaluated at the global maximum likelihood estimate, $\hat\theta$,  
thus $\hat{\jmath} =  j(\hat\psi, \hat\lambda) = j(\hat\theta)$.

We use $d/d\psi$ to denote the total derivative with respect to $\psi$ and $\partial/\partial\psi$ to denote the partial derivative with respect to $\psi$.
The $k$-th derivative of the profile log-likelihood is $\zeta_{k} (\psi) = d^k l_p(\psi)/d \psi^k$ and $k$-th total derivative of the log determinant of the information matrix is 
\[ \gamma_k(\psi) =  d^k \log \lbrace|\info (\psi, \hat \lambda_{\psi})| \rbrace/d\psi^k.\]
We define the $k$-th quasi-cumulant of the profile log-likelihood function as:
\begin{align}
    \kappa_k(\psi) = \frac{\zeta_k(\psi)}{\lbrace -\zeta_2(\psi)  \rbrace^{k/2}  }. \label{eq:kappa}
\end{align}  
In the sequel we suppress the dependence of functions on the data $y$ in the notation, for example $r(\psi_0;  y)$ will simply be written as $r(\psi_0)$. Finally $\sigma_{\max}(A)$ is the maximum singular value of a matrix $A$. The following inequalities will prove useful, for two square matrices $A, B$ and $B$ positive definite:
\begin{align*}
|\tr(AB)| \leq \norm{A}_F \norm{B}_F \quad \text{and,} \quad |\tr(AB)| \leq \sigma_{\max}(A) \tr(B) ,
\end{align*}
where $\norm{A}_F = (\sum_{i,j} A_{ij}^2)^{1/2}$ is the Frobenius norm and the latter inequality is a consequence of the von Neumann trace inequality \citep{trace}. For a vector $v$, we let $\norm{v}_p$ denote its $L^p$ norm.

\section{Background and Assumptions}\label{Sec2}
We assume the following conditions on the model:

\begin{assumption}\label{ass:smooth}
The $k$-th order partial derivatives of the log-likelihood function with respect to the elements of $\theta$ are $O_p(n)$ for all integers $k > 1$ when evaluated at $\hat\theta$.
\end{assumption}

\begin{assumption}\label{ass:consistency}
$\hat\theta - \theta_0 = O_p(n^{-1/2})$
\end{assumption}
\begin{assumption}\label{ass:eigen} 
The eigenvalues of $\hat\jmath/n$, $j(\hat{\theta}_{\psi_0})/n$, $ n \{\hat\jmath\}^{-1}$ and $ n j^{-1}(\hat{\theta}_{\psi_0})$ are positive and $O_p(1)$.
\end{assumption}

Assumption \ref{ass:smooth} ensures that the likelihood derivatives grow at the usual rate when evaluated at the maximum likelihood estimate. Assumption \ref{ass:consistency} states that the maximum likelihood estimate is $n^{1/2}$ consistent for the true parameter value: this rate of consistency is typically achieved for most well behaved parametric models \citep[~\S5]{vaart_1998}. Finally Assumption \ref{ass:eigen} ensures that the asymptotic covariance matrix for the maximum likelihood estimate is well defined. 
For the regression problems that we consider, Assumption \ref{ass:eigen} is equivalent to a restriction on the eigenvalues on the Gramian matrix, $X^\top X$,  as well as a lower bound on the variance of the fitted values.

\subsection{Linear Exponential Family}
Let $X$ be a $n \times p$ matrix of covariates, $x_{i,j}$ denote the $(i,j)^{th}$ entry in $X$, and $x^\top_i$ denote the $i$-th row of $X$. 
We assume the density of $y_i$  is a full exponential family model with log-likelihood function
\begin{align}
      \log\lbrace f(y_i; \psi, \lambda, x_i) \rbrace =  \psi u(x_{i,p},y_i) + \lambda^\top v(x_i,y_i) - c_i(\psi,\lambda) + h(x_i,y_i)  \label{exp},
 \end{align}
 where $u(x_{i,p},y_i)$ is a scalar sufficient statistic associated with $\psi$ and 
 \[v(x_i,y_i) = \lbrace v_1(x_{i,1}, y_i), \cdots, v_{p-1}(x_{i,p-1}, y_i) \rbrace^\top,\]
 is the vector of sufficient statistics associated with the nuisance parameters.
 
In this model, $q$ in (\ref{eq:intro+rstar}) takes the form
\begin{align*}
q(\psi_0)  = t(\psi_0)\rho(\psi_0), 
\end{align*}
where
\begin{align}
t &=  (\hat\psi - \psi_0) j_{\text{p}}^{1/2}(\hat\psi), \label{wald_stat}
\end{align}
where $t$ is the Wald statistic for testing $\psi = \psi_0$,
\begin{align*}
\rho &= \{|j_{\lambda\lambda}(\hat{\theta})|/|j_{\lambda\lambda}(\hat\theta_{\psi_0})|\}^{1/2}, 
\end{align*}
where $j_{\lambda\lambda}(\theta)$ is the $(p - 1) \times (p-1)$ sub-matrix of $j(\theta)$ associated with the nuisance parameters \citep[~\S8.6.1]{2007applied}. 
We follow \cite{peters} and write
\begin{align*}
    r^\star =r  + r_{np} + r_{inf},
\end{align*}
where $r_{np}$ is a nuisance parameter adjustment and $r_{inf}$ is an information adjustment.
This partitioning of the adjustments will be helpful to the analysis of the asymptotic behaviour of $r^\star$.
For this model,
\begin{align}
 r_{np} =  \frac{1}{r} \log(\rho) \text{,} \quad 
r_{inf} &= \frac{1}{r} \log(\frac{t}{r}). \label{eq:adj_def}
\end{align}

\subsection{Location-Scale Family}
For a linear regression model based on the location scale-family, the model is
\begin{align}
    y_i =  x^\top_i \beta + \sigma\epsilon_i, \label{locscale}
\end{align}
where the errors $\epsilon_i$ are assumed independent and identically distributed from a known distribution with continuous density $f(\epsilon)$. 
The model is parametrized by $ \theta = (\beta, \sigma)$, we assume that the parameter of interest is a component of $\beta$.
For this model 
\begin{align}
q(\psi_0) &= s(\psi_0)/\rho(\psi_0)\nonumber\\ 
s &=  l_{\text{p}}^\prime(\psi_0)/ j_{\text{p}}^{1/2}(\mle), \label{eq:score}
\end{align}
$s$ is the standard score test statistic and $\rho$ is defined above. For this model
\[\quad r_{np} = - \frac{1}{r} \log(\rho), \quad  r_{inf} = \frac{1}{r} \log(\frac{s}{r}). \]

\subsection{General Models}
For general models $r_{np}$ and $r_{inf}$ are more difficult to work with, the additional difficulty lies in the necessity of conditioning on an ancillary statistic.
It may be possible to use certain variants of $r^\star$, where it is written in a similar form to a Wald or a Score statistic \cite[~3.4]{reid2003} and proceed in a similar fashion as in proof of Theorem \ref{th:1} and \ref{th:2}, we do not pursue this however. 


\section{Formal expansions of $r_{inf}$ and $r_{np}$} \label{Sec3}
In this section we obtain formal asymptotic expansions for $r_{inf}$ and $r_{np}$, which detail the relationship between $r$ and $r^\star$ in the linear exponential and location scale families, respectively. Analogous results to Theorem \ref{th:1} and \ref{th:2} were obtained by \cite{CAKMAK1998211} in the case of no nuisance parameters. The expansions for $r_{inf}$ and $r_{np}$ show that $r^\star$ is asymptotically equivalent to a location and scale adjustment to the likelihood root $r$. 

We begin by showing the relationships between $r$, $s$ and $t$.

\begin{lemma}\label{Lemma:1}
Under Assumptions \ref{ass:smooth}--\ref{ass:eigen}, for $r$, $t$ and $s$ defined in (\ref{eq:r}), (\ref{wald_stat}) and (\ref{eq:score}):
\begin{align*}
t = r\left\{ 1 + \frac{A_1}{n^{1/2}} r +  \frac{B_1}{n} r^2 + O_p(n^{-3/2}) \right\}, \\
s = t\left\{ 1 + \frac{A_2}{n^{1/2} }t + \frac{B_2}{n }t^2 + O_p(n^{-3/2})  \right\},  
\end{align*}
where, 
\begin{align*}
A_1 &= - \frac{n^{1/2}}{6} \kappa_3(\hat\psi), \quad B_1 =  \frac{n}{24}\kappa_4(\hat\psi) + \frac{5n}{72} \kappa_3^2(\hat\psi), \\
A_2 &= \frac{n^{1/2}\kappa_3(\hat\psi)}{2}, \quad\quad  B_2 = -\frac{n\kappa_4(\hat\psi)}{6}.
\end{align*}
\end{lemma}

\begin{proof}
We begin by deriving the relationship between $r$ and $t$, defined in (\ref{eq:r}) and (\ref{wald_stat}) :
\begin{align*}
r^2 &= 2\Blbrace l_{\text{p}}(\hat\psi) - l_{\text{p}}(\psi_0) \Brbrace,\\
&= 2\Blbrace l_{\text{p}}(\hat\psi) - l_{\text{p}}(\hat\psi) + ( \hat\psi - \psi_0)\zeta_1(\mle) - \frac{( \mle - \psi_0)^2}{2} \zeta_2(\mle) \\
&\quad + \frac{( \mle - \psi_0)^3}{6} \zeta_3(\mle) + \frac{( \mle - \psi_0)^4}{24} \zeta_4(\mle) +  +O_p(n^{-3/2}) \Brbrace\\
&=t^2 \Blbrace  1 + \frac{\kappa_3(\hat\psi )}{3}t  - \frac{\kappa_4(\hat\psi)}{12}t^2  + O_p(n^{-3/2}) \Brbrace.
\end{align*}
The Taylor-series expansion for $(1 + x)^{1/2}$ gives
\begin{align*}
r &= t\Blbrace  1 + \frac{ \kappa_3(\hat\psi )}{6}t  - \frac{ \kappa_4(\hat\psi)}{24}t^2   - \frac{\kappa^2_3(\hat\psi )}{72}t^2  +O_p(n^{-3/2})\Brbrace,
\end{align*}
which implies
\begin{align}
t &= r \Blbrace  1 + \frac{\kappa_3(\hat\psi )}{6} t  - \frac{ \kappa_4(\hat\psi)}{24} t^2  - \frac{\kappa^2_3(\hat\psi )}{72}t^2  + O_p(n^{-3/2})\Brbrace^{-1}\nonumber \\
&= r \Blbrace 1 - \frac{\kappa_3(\hat\psi )}{6}t  + \frac{\kappa_4(\hat\psi)}{24} t^2  + \frac{\kappa^2_3(\hat\psi )}{72}t^2  + \frac{\kappa^2_3(\hat\psi )}{36}t^2   +O_p(n^{-3/2}) \Brbrace \nonumber \\
&=  r \Blbrace 1 - \frac{\kappa_3(\hat\psi )}{6}t  + \frac{\kappa_4(\hat\psi)}{24}t^2   + \frac{\kappa^2_3(\hat\psi )}{24}t^2  +O_p(n^{-3/2}) \label{r_inf_1} \Brbrace .
\end{align} 
As $t$ appears on both sides of the equation, we iteratively solve the equation by substitution. 
\begin{align}
t &= r\Big[  1 - \frac{\kappa_3(\hat\psi )}{6}   \Blbrace \Blp   1 - \frac{\kappa_3(\hat\psi )}{6}t   \Brp r \Brbrace + \frac{1}{24} \Blbrace \kappa_3^2(\hat\psi) + \kappa_4(\hat\psi) \Brbrace r^2    +O_p(n^{-3/2}) \Big]\label{sub_1} \\
&= r \Big[  1 - \frac{\kappa_3(\hat\psi )}{6}r   + \Blbrace \frac{\kappa^2_3(\hat\psi)}{36}r  \Brbrace t   + \frac{1}{24} \Blbrace \kappa_3^2(\hat\psi) + \kappa_4(\hat\psi) \Brbrace  r^2 +O_p(n^{-3/2}) \Big] \nonumber \\
&= r \Blbrace  1 - \frac{ \kappa_3(\hat\psi )}{6} r   +  \frac{5}{72} \kappa^2_3(\hat\psi) r^2    + \frac{1}{24}  \kappa_4(\hat\psi)r^2  +O_p(n^{-3/2}) \Brbrace \nonumber \\
&= r\Blbrace 1 + \frac{A_1}{n^{1/2}}r + \frac{B_1}{n}r^2  +O_p(n^{-3/2}) \Brbrace, \label{sub_2}
\end{align}
where
\begin{align*}
A_1 =- \frac{n^{1/2}}{6} \kappa_3(\hat\psi), \quad B_1 =  \frac{n}{24}\kappa_4(\hat\psi) + \frac{5n}{72} \kappa_3^2(\hat\psi).
\end{align*}
While for the expansion of $s$, write
\begin{align*}
   s &= \frac{\zeta_1(\psi_0)}{j_p^{1/2}(\mle)} \\
   &= \frac{1}{j_p^{1/2}(\mle) } \Blbrace \zeta_1(\mle) - \zeta_2(\mle)(\mle - \psi_0) +\frac{\zeta_3(\mle)}{2}( \mle - \psi_0)^2 - \frac{\zeta_4(\mle)}{6}( \mle - \psi_0)^3 + O_p(n^{-3/2})\Brbrace\\
    &= t +\frac{\kappa_3(\mle) }{2}t^2 - \frac{\kappa_4(\mle) }{6}t^3 + O_p(n^{-3/2}),\\
    &= t \left\{ 1 + \frac{A_2}{n^{1/2}}t + \frac{B_2}{n}t^2 +O_p(n^{-3/2}) \right\},
\end{align*}
where
\begin{align*}
A_2 = \frac{n^{1/2}}{2} \kappa_3(\mle), \quad B_2 =  -\frac{n}{6}\kappa_4(\mle).
\end{align*}
\end{proof}

\begin{theorem}\label{th:1}
Under Assumptions \ref{ass:smooth}--\ref{ass:eigen},
for the linear exponential family
\begin{align}
r_{np} =  \frac{1}{2}  \frac{\gamma_1(\mle)}{j_{\normalfont\text{p}}(\mle)^{1/2}} - \Blbrace  \frac{1}{12}   \frac{\kappa_3(\mle) \gamma_1(\mle)}{j_{\normalfont\text{p}}(\mle)^{1/2}} - \frac{1}{4}    \frac{\gamma_2(\mle) }{ j_{\normalfont\text{p}}(\mle)}\Brbrace r  
+O_p(n^{-3/2}), \label{eq:exp_np}
 \end{align}
and for the location-scale family 
\begin{align}
r_{np} = -\frac{1}{2}  \frac{\gamma_1(\mle)}{j_{\normalfont\text{p}} (\mle)^{1/2}} + \Blbrace  \frac{1}{12}   \frac{\kappa_3(\mle) \gamma_1(\mle)}{j_{\normalfont \text{p}}(\mle)^{1/2}} - \frac{1}{4}    \frac{\gamma_2(\mle) }{ j_{\normalfont\text{p}}(\mle)}\Brbrace r  
+O_p(n^{-3/2}). \label{eq:ls_np}
\end{align} 
\end{theorem}

\begin{theorem}\label{th:2}
Under Assumptions \ref{ass:smooth}--\ref{ass:eigen},
for the linear exponential family, 
\begin{align}
r_{inf} = -\frac{1}{6}\kappa_3(\mle) + \Blbrace \frac{1}{24}\kappa_4(\mle) + \frac{4}{72}\kappa_3^2(\mle) \Brbrace r  +O_p(n^{-3/2}),  \label{eq:exp_inf}
\end{align}

and for the location-scale family 
\begin{align}
r_{inf} = \frac{1}{3}\kappa_3(\mle) - \Blbrace \frac{3}{24}\kappa_4(\mle) + \frac{11}{72}\kappa_3^2(\mle) \Brbrace r  +O_p(n^{-3/2}) . \label{eq:ls_inf}
\end{align}

\end{theorem}

\begin{proof}
\textbf{Linear Exponential Family:}
Using (\ref{eq:adj_def}) and Lemma \ref{Lemma:1} for a linear exponential family we have
\begin{align}
r_{inf} &= \frac{1}{r} \log\Blp \frac{t}{r } \Brp  \nonumber \\
&= \frac{A_1}{n^{1/2}} + \left(\frac{B_1}{n} - \frac{A_1^2}{2n} \right)r  +O_p(n^{-3/2}) \nonumber \\
&= -\frac{1}{6}\kappa_3(\mle) + \Blbrace \frac{1}{24}\kappa_4(\mle) + \frac{4}{72}\kappa_3^2(\mle) \Brbrace r  +O_p(n^{-3/2}) \label{r_inf}.
\end{align}
A similar expansion can be developed for $r_{np}$:
\begin{align}
r_{np} &=  \frac{1}{2r} \log\Blbrace  \frac{|j_{\lambda\lambda}(\hat\psi, \hat\lambda)| }{ |j_{\lambda\lambda}(\psi_0, \hat\lambda_{\psi_0} ) |} \Brbrace \nonumber \\
&=\frac{1}{2r} \left[ \frac{\gamma_1(\mle)}{ \lbrace-\zeta_2(\mle)\rbrace^{1/2} } t+  \frac{\gamma_2(\mle)}{2\zeta_2(\mle)} t^2 + O_p\Blp n^{-3/2} \Brp \right]  \label{eq:sub_step} \\
& = \frac{1}{2} \left[ \Blp 1 + \frac{A_1}{n^{1/2}} r \Brp \frac{\gamma_1(\mle)}{\lbrace-\zeta_2(\mle)\rbrace^{1/2} } + \frac{\gamma_2(\mle)}{2\zeta_2(\mle)}r + O_p\Blp n^{-3/2} \Brp \right] \nonumber \\
&=  \frac{1}{2} \frac{\gamma_1(\mle)}{ \lbrace -\zeta_2(\mle) \rbrace^{1/2}} + \left[ \frac{1}{2} \frac{A_1\gamma_1(\mle)}{ \lbrace -\zeta_2(\mle) n \rbrace^{1/2} } + \frac{\gamma_2(\mle)}{4\zeta_2(\mle)}  \right] r  + O_p\Blp n^{-3/2}\Brp \label{r_np},
\end{align}
where the third equality uses Lemma \ref{Lemma:1}.

\textbf{Location-Scale Family:}
Using Lemma \ref{Lemma:1}, we obtain that for the location-scale family:
\begin{align*}
r_{inf} &= \frac{1}{r} \log\Big[ \frac{1}{r } \Blbrace t + \frac{\kappa_3(\mle) }{2}t^2 -  \frac{\kappa_4(\mle) }{6}t^3 + O_p(n^{-3/2}) \Brbrace  \Big]\\
&= \frac{1}{r} \log\Big[ 1 + \frac{A_2}{n^{1/2}}r + \frac{B_2}{n} r^2 + \frac{\kappa_3(\mle)}{2}  \Blbrace 1 + \frac{A_2}{n^{1/2}}r \Brbrace^2 r - \frac{\kappa_4(\mle) }{6}r^2  + O_p(n^{-3/2})  \Big] \\
&= \frac{1}{3}\kappa_3(\mle) - \Blbrace\frac{3}{24}\kappa_4(\mle) + \frac{11}{72}\kappa_3^2(\mle)\Brbrace r + O_p(n^{-3/2}) .
\end{align*}
The expansion for the nuisance parameter adjustment $r_{np}$ is the same as in the exponential family, except for a change in sign.
   
\end{proof}

\begin{remark} \label{Remark:linearity}
From Theorems \ref{th:1} and \ref{th:2}, for the linear exponential family, 
\begin{align*}
r^\star &= -\frac{1}{6}\kappa_3(\mle) + \frac{1}{2}  \frac{\gamma_1(\mle)}{\lbrace- \zeta_2(\mle)\rbrace^{1/2}} 
\\
&+ \left[ 1 + \frac{1}{24}\kappa_4(\mle) + \frac{4}{72}\kappa_3^2(\mle) -   \frac{1}{12}\frac{\kappa_3(\mle) \gamma_1(\mle)}{ \lbrace - \zeta_2(\mle) \rbrace^{1/2}} - \frac{1}{4}    \frac{\gamma_2(\mle)}{ \zeta_2(\mle)}\right] r  
+O_p(n^{-3/2}),
\end{align*}
and for the location-scale family,
\begin{align*}
r^\star &= \frac{1}{3}\kappa_3(\mle) - \frac{1}{2}  \frac{\gamma_1(\mle)}{\lbrace- \zeta_2(\mle)\rbrace^{1/2}} 
\\
&+ \left[ 1 - \frac{3}{24}\kappa_4(\mle) - \frac{11}{72}\kappa_3^2(\mle) + \frac{1}{12}\frac{\kappa_3(\mle) \gamma_1(\mle)}{ \lbrace - \zeta_2(\mle) \rbrace^{1/2}} + \frac{1}{4}    \frac{\gamma_2(\mle)}{ \zeta_2(\mle)}\right] r  
+O_p(n^{-3/2}).
\end{align*}
After some algebraic manipulations using the above equations, we obtain for both families:
\begin{align*}
r^\star =  \frac{ r - \tilde{A}/n^{1/2}}{ \{ 1 + \tilde{B}/n \}} + O_p(n^{-3/2}), 
\end{align*}
which shows that $r^\star$ is a shift and rescaling of $r$, where $\tilde{A}$ and $\tilde{B}$ are $O_p(1)$ terms.

\end{remark}

\begin{corollary}
Under Assumptions \ref{ass:smooth}--\ref{ass:eigen},
the expansions of $r_{inf}$ and $r_{np}$ can be extended to:
\begin{align} 
 r_{inf} =  \sum_{k = 1}^m \frac{A_{k}}{n^{k/2}} r^{k-1} + O_p(n^{-(m+1)/2}), \quad r_{np} = \sum_{k = 1}^m \frac{B_k}{n^{k/2}} r^{k-1} + O_p(n^{-(m+1)/2}), \label{eq:cor_1}
\end{align}
for arbitrary $m \in \mathbb{N}$
where the coefficients $A_k = O_p(1)$ and $B_k = O_p(1)$.  
\end{corollary}

From the substitution argument employed in (\ref{sub_1}) to (\ref{sub_2}), we can deduce that the coefficient of $r^{k-1}$ are of order $O_p( n^{-k/2} )$ and 
\[A_k = n^{k/2}\sum_{j = 1}^k \sum_{\lbrace i_1, \cdots, i_j \rbrace \in S_j}  Z_{j} (\lbrace i_1, \cdots, i_j \rbrace) \prod_{l = 1}^j   \hat\kappa_{i_l} , \] 
where indices $\lbrace i_1, \cdots, i_j \rbrace$ take values in  $ \lbrace 3, \cdots, k + 2 \rbrace$, $S_j$ is the set of all indices such that $\sum_{l = 1}^j (i_l - 2) = k$, and $Z_{j}(\cdot)$ is a function of a set of indices that returns a numerical constant. 

From equation (\ref{r_np}), when grouping terms in powers of $r$,
we obtain the following expressions for the coefficients in the expansion of $r_{np}$:
\[ B_k = \sum_{j = 0}^{k - 1} n^{(k-j)/2} C_j \frac{\gamma_{k - j}(\hat\psi)}{ \lbrace-\zeta_2(\mle)\rbrace^{(k-j)/2} } \quad \text{ and } \quad C_k = \sum_{m = 1}^k \sum_{\lbrace i_1, \cdots, i_m \rbrace \in D_{m,k}} \prod_{l = 1}^m A_{i_l} ,\]
where $C_0 = 1$, the indices $\lbrace i_1, \cdots, i_m\rbrace$ range from $1$ to $j$. The set $D_{m,k}$ is the set of all indices $\lbrace i_1, \cdots, i_m \rbrace$ such that $\sum_{l = 1}^m i_l = k$. Some examples of terms which appear in $A_k$, $B_k$ and $C_k$ are given in Figure \ref{table:coef}.

\begin{figure}[h!]
\begin{tabular}{ |p{0.5cm}||p{4.5cm}|p{4.5cm}|p{4.5cm}|  }
 \hline
 $k$     & $A_k$ & $B_k$ & $C_k$ \\
 \hline
 1   & $n^{1/2}\kappa_3(\hat\psi)$    & $\frac{n^{1/2}\gamma_1(\hat\psi)}{ \{- \zeta_2(\hat\psi)\}^{1/2} }$ & $A_1$\\
 \hline 
 2&   $n\kappa_3^2(\hat\psi)$, $n\kappa_4(\hat\psi)$   &  $\frac{n^{1/2} C_1\gamma_1(\hat\psi)}{ \{- \zeta_2(\hat\psi)\}^{1/2} }$, $\frac{n\gamma_2(\hat\psi)}{ - \zeta_2(\hat\psi)\ }$   & $A_2$, $A_1^2$\\
 \hline 
 3 &  $n^{3/2}\kappa_3^3(\hat\psi)$, $n^{3/2}\kappa_3(\hat\psi)\kappa_4(\hat\psi)$ , $n^{3/2}\kappa_5(\hat\psi)$  & $\frac{C_2n^{1/2}\gamma_1(\hat\psi)}{ \{- \zeta_2(\hat\psi)\}^{1/2} } $, $\frac{C_1n\gamma_2(\hat\psi)}{ - \zeta_2(\hat\psi) }$, $\frac{n^{3/2}\gamma_3(\hat\psi)}{ \{- \zeta_2(\hat\psi)\}^{3/2} }$  &  $A_3$,$A_1A_2$, $A_1^3$  \\
 \hline
 4   & $n^2 \kappa_3^4(\hat\psi)$, $n^2\kappa^2_3(\hat\psi)\kappa_4(\hat\psi)$ , $n^2\kappa^2_4(\hat\psi)$, $n^2\kappa_5(\hat\psi) \kappa_3(\hat\psi)$, $n^2\kappa_6(\hat\psi)$  &$\frac{C_3n^{1/2}\gamma_1(\hat\psi)}{ \{- \zeta_2(\hat\psi)\}^{1/2} } $, $\frac{C_2n\gamma_2(\hat\psi)}{ \{- \zeta_2(\hat\psi)\} } $, $\frac{C_1n^{3/2}\gamma_3(\hat\psi)}{ \{- \zeta_2(\hat\psi)\}^{3/2} }$, $\frac{n^2\gamma_4(\hat\psi)}{ \{- \zeta_2(\hat\psi)\}^{2} }$&  $A_4$, $A_3A_1$, $A_2^2$, $A_1^2$, $A_2$, $A_1^4$ \\
 \hline
\end{tabular}
\caption{The first four terms of $A_k$, $B_k$ and $C_k$ for reference. The order of the $\kappa_j$ terms are given in Lemma 2 in \S\ref{Appendix}.1.} 
\label{table:coef}
\end{figure}

\section{$r^\star$ in High Dimensions}\label{Sec4}
We discuss the behaviour of $r^\star$ in the high dimensional setting, when $p$ increases with $n$.
Recently this asymptotic paradigm has become increasingly popular as the datasets observed in practice are increasing in size not only in terms of number of samples but also in number of covariates. 
It is known that some traditional statistics, such as the likelihood ratio test, does not perform well in this setting \citep{SurCandesPNAS}. Thus, it is reasonable to question how $r^\star$ performs in this setting. 
We quantify the behaviour of the order of the adjustment factors $r_{inf}$ and $r_{np}$ within this asymptotic regime, and show that the order is potentially much larger than in the $p$-fixed asymptotic regime. 
A similar analysis is performed in \cite{Tang_Reid}, studying the behaviour of $r_{inf}$ and $r_{np}$, although the approach differs from the one taken here as it relies on a direct expansion of the adjustment factor instead of the iterative substitution argument employed.
The results obtained from the two approaches coincide in the two families considered. 

In addition to the assumptions made in \S3, we assume that $j_{\psi\lambda}(\theta) = O_p(n^{1/2})$ for all $\theta \in \lbrace \theta : \norm{\theta - \theta_0}_2 < \delta \rbrace$ for some $\delta > 0$. This will be possible if the parametrization $(\psi, \lambda)$ is orthogonal in the \cite{adj_profile} sense. 
We also assume that the $k$-th  order derivatives of the constrained maximum likelihood estimate are $O_p(1)$ for $k \geq 2$.
We also place certain restrictions on the sizes of the eigenvalues of the third likelihood derivative matrices, these restrictions vary by the family we consider. 

\subsection{Linear Exponential Family}
In the linear exponential family, there exists an orthogonal parametrization $\hat{\lambda}_\psi = \hat\lambda$, so that under this parametrization $d\hat\lambda_\psi^k/d\psi^k = 0$ for for all integer-valued $k$ as the observed and expected information coincide. For this sub-section assume that the maximum singular value of $j_{\psi\lambda\lambda}(\hat\theta)$ is $O_p(n)$.
First consider the leading terms which appear in $r_{np}$ and $r_{inf}$ in (\ref{eq:exp_np}) and (\ref{eq:exp_inf}). We make the strong assumption that the higher-order terms are smaller. We find that, $\kappa_3(\hat\psi)$ and $\kappa_4(\hat\psi)$, as defined in (\ref{eq:kappa}), satisfy 

\[ \kappa_3(\hat\psi) = O_p(n^{-1/2}), \quad \kappa_4(\hat\psi) = O_p(n^{-1}),\]
which implies that
\[r_{inf} = O_p ( n^{-1/2}  ).\]

For $r_{np}$,  we have
\begin{align*}
\left| \gamma_1(\hat\psi) \right| = \left| \tr\left[ j_{\lambda\lambda}^{-1}(\hat\psi) j_{\psi\lambda\lambda}(\hat\psi) \right] \right| \leq \sigma_{\max}\lbrace j_{\psi\lambda\lambda}(\hat\theta) \rbrace \tr[j_{\lambda\lambda}^{-1}(\hat\theta)]  = O_p(p),
\end{align*}
using von Neumann's inequality and
\begin{align*}
\left|\gamma_2(\hat\psi)\right| &= \left| \tr\left[ j_{\lambda\lambda}^{-1}(\hat\psi) j_{\psi\lambda\lambda}(\hat\psi)j_{\lambda\lambda}^{-1}(\hat\psi) j_{\psi\lambda\lambda}(\hat\psi) \right] + \tr\left[  j_{\lambda\lambda}^{-1}(\hat\psi) j_{\psi\psi\lambda\lambda}(\hat\theta) \right] \right| \\
&\leq \sigma^2_{\max}\lbrace j_{\psi\lambda\lambda}(\hat\theta) \rbrace \tr[ \lbrace j_{\lambda\lambda}^{-1}(\hat\theta) \rbrace^2] +\norm{j^{-1}_{\lambda\lambda}(\hat\theta)}_F \norm{j_{\psi\psi\lambda\lambda}(\hat\theta)}_F\\
&\leq \sigma^2_{\max}\lbrace j_{\psi\lambda\lambda}(\hat\theta) \rbrace p \sigma_{\max}[\lbrace j_{\lambda\lambda}^{-1}(\hat\theta) \rbrace^2 ] + p^{1/2} \sigma_{\max}\lbrace j_{\lambda\lambda}^{-1}(\hat\theta) \rbrace \norm{j_{\psi\psi\lambda\lambda}(\hat\theta)}_F \\
&= O_p(p) + O_p(p^{3/2}).
\end{align*}
In the above we have used von Neumann's trace inequality, and the fact that the largest singular value is bounded by the Frobenius norm. This implies that
\[ r_{np} = O_p\left( pn^{-1/2} \right), \]
from Theorem \ref{th:1}.

If $p = o(n^{1/2})$, these results are consistent with \cite{Tang_Reid}. 
The above implies that $r$ coincides with $r^\star$ asymptotically in distribution if $p =o(n^{1/2})$.

\subsection{Location-Scale Family}
We only consider the leading terms which appear in $r_{inf}$ and $r_{np}$ in (\ref{eq:ls_inf}) and (\ref{eq:ls_np}), making the strong assumption that the other terms are of smaller order, as was done for the exponential family.
Under the orthogonal parametrization \citep[Lemma 1]{Tang_Reid}:
\[ \norm{\frac{d\hat\lambda_\psi}{d\psi}|_{\psi = \hat\psi} }_2 = O_p(p^{1/2}n^{-1/2}), \]
showing that $\kappa_3(\hat\psi) = O_p(p/n^{1/2})$, and $\kappa_4(\hat\psi) = O_p(p^2/n)$, which further implies that 
\[r_{inf} =  O_p\lbrace(  \max(p/n^{1/2}, p^2/n) \rbrace ,\] 
showing a dependence in $p$ not present in the linear exponential family.
Next we examine the size of $r_{np}$. As the derivatives of the constrained maximum likelihood estimate are not 0, we make the assumptions that  $\max_{i = 1, \cdots, p} \sigma_{\max}\{j_{\theta_i\lambda\lambda}(\hat\theta)\} = O_p(n)$. 
Then
\begin{align*}
|\gamma_1(\hat\psi)| &= \left|\tr\left[ j_{\lambda\lambda}^{-1}(\hat\theta) \frac{d}{d\psi} j_{\lambda\lambda}(\hat\theta_\psi)|_{\psi = \hat\psi} \right]\right| \leq \sigma_{\max} \left\lbrace \frac{d}{d\psi} j_{\lambda\lambda}(\hat\theta_\psi)|_{\psi = \hat\psi} \right\rbrace \tr[j_{\lambda\lambda}^{-1}(\hat\theta)] \\
&= \sigma_{\max}\left\lbrace j_{\psi\lambda\lambda}(\hat\theta) + \sum_{i = 1}^{p-1} \frac{\partial\hat\lambda_{\psi, i }}{\partial\psi} |_{\psi = \hat\psi} \ j_{\lambda_i\lambda\lambda}(\hat\theta) \right\rbrace \tr[j_{\lambda\lambda}^{-1}(\hat\theta)] \\
&\leq \left[ \sigma_{\max}\left\lbrace j_{\psi\lambda\lambda}(\hat\theta) \right\rbrace + \norm{ \frac{\partial\hat\lambda_{\psi, i }}{\partial\psi} |_{\psi = \hat\psi} }_1 \ \sigma_{\max}\left\lbrace  j_{\lambda_i\lambda\lambda}(\hat\theta) \right\rbrace \right] \tr[j_{\lambda\lambda}^{-1}(\hat\theta)]\\
&\leq \left[ \sigma_{\max}\left\lbrace j_{\psi\lambda\lambda}(\hat\theta) \right\rbrace + p^{1/2}\norm{\frac{\partial\hat\lambda_{\psi, i }}{\partial\psi} |_{\psi = \hat\psi}}_{2} \sigma_{\max}\left\lbrace  j_{\lambda_i\lambda\lambda}(\hat\theta) \right\rbrace \right] \tr[j_{\lambda\lambda}^{-1}(\hat\theta)]\\
&=   \lbrace O_p(n) +O_p(p n^{1/2}) \rbrace O_p(p/n) = O_p\left\{ \max (p, p^2/n^{1/2} ) \right\}.
\end{align*}
Also,
\begin{align*}
\left|\gamma_2(\hat\psi)\right| &= \left| \tr\left[  j_{\lambda\lambda}^{-1}(\hat\theta) \frac{d}{d\psi} j_{\lambda\lambda}(\hat\theta_\psi)|_{\psi = \hat\psi} \ j_{\lambda\lambda}^{-1}(\hat\theta) \frac{d}{d\psi} j_{\lambda\lambda}(\hat\theta_\psi)|_{\psi = \hat\psi} \right] + \tr\left[  j_{\lambda\lambda}^{-1}(\hat\theta) \frac{d^2}{d\psi^2} j_{\lambda\lambda}(\hat\theta_\psi)|_{\psi = \hat\psi} \right] \right| \\
&\leq \sigma^2_{\max}\left\lbrace  \frac{d}{d\psi} j_{\lambda\lambda}(\hat\theta_\psi)|_{\psi = \hat\psi} \right\rbrace \tr[ \lbrace j_{\lambda\lambda}^{-1}(\hat\theta) \rbrace^2]  +\tr[j^{-1}_{\lambda\lambda}(\hat\theta)] \sigma_{\max} \left\lbrace\frac{d^2}{d\psi^2} j_{\lambda\lambda}(\hat\theta_\psi)|_{\psi = \hat\psi} \right\rbrace\\
&= O_p\lbrace \max (p^2, p^4/n)\rbrace + O_p\lbrace \max( pn, p^2n^{1/2}, p^3)\rbrace.
\end{align*}
Detailed calculation of the maximum singular value of $d^2j_{\lambda\lambda}(\hat\theta_\psi)/d\psi^2 |_{\psi = \hat\psi}$ is deferred to \S\ref{appendix:max} 
If $p = o(n^{1/2})$, then
\begin{align*}
r_{np} = O_p(p/n^{1/2}),
\end{align*}
which agrees with the result in \citet[~\S5.2]{Tang_Reid}.

\section{$r$ as a series in $q$}\label{Sec5}
In this section, we obtain an expansion of $r$ as a polynomial in $q$, which is used to justify the normal approximation to the distribution of $r^\star$, see \citet[~\S 8]{2007applied}. 
We do not give the exact forms of the coefficients which appear in the polynomial expansion, as they are not as simple as those in \S\ref{Sec3}. 

\begin{theorem}
Under Assumptions \ref{ass:smooth}--\ref{ass:eigen},
for the linear exponential and location-scale families,
\[ r = q + \frac{A}{n^{1/2}}q^2 + \frac{B}{n}q^3 + O_p\left( n^{-3/2}  \right), \]
where $A$ and $B$ are $O_p(1)$.
\end{theorem}

\begin{proof}
Note that $\rho^{-1} = O_p(1)$ by Assumption \ref{ass:eigen}. 
We prove the result in the case for the linear exponential family; the proof for the location-scale family is similar. 
We use capital letters to denote terms of order $O_p(1)$. From (\ref{r_inf_1}),
\begin{align}
t &= r \Blbrace 1 - \frac{\kappa_3(\hat\psi )}{6}t  + \frac{\kappa_4(\hat\psi)}{24} t^2  + \frac{\kappa^2_3(\hat\psi )}{24} t^2 +O_p(n^{-3/2}) \Brbrace\nonumber \\
&=  r \left[ 1 -  \frac{\kappa_3(\hat\psi )}{6\rho} q +   \left\lbrace \frac{1}{24\rho^2} \kappa_4(\hat\psi)  + \frac{1}{24\rho^2} \kappa^2_3(\hat\psi )\right\rbrace q^2 +O_p(n^{-3/2}) \right] \nonumber \\
&= r\left\lbrace 1 + \frac{C}{n^{1/2}}q + \frac{D}{n}q^2 + O_p(n^{-3/2})\right\rbrace. \label{eq:t_expansion}
\end{align} 
We expand $|j_{\lambda\lambda} (\hat\theta_{\psi_0})|$,
\begin{align*}
|j_{\lambda\lambda} (\hat\theta_{\psi_0})| &=  |j_{\lambda \lambda} (\hat\theta)| +  (\psi_0 - \hat\psi) \frac{d|j_{\lambda \lambda} (\hat\theta_{\psi_0})|}{d\psi}|_{\psi = \hat\psi} + \frac{1}{2} (\psi_0 - \hat\psi)^2 \frac{d|j_{\lambda \lambda} (\hat\theta_{\psi_0})|}{d\psi^2} |_{\psi =\hat\psi} + \cdots  \\
&= |j_{\lambda \lambda} (\hat\theta)| \left\lbrace 1 + (\hat\psi - \psi_0 )  \gamma_1(\hat\psi) + (\hat\psi - \psi_0 )^2 \gamma_2(\hat\psi)  + O_p(n^{-3/2}) \right\rbrace \\
&= |j_{\lambda \lambda} (\hat\theta)| \left[ 1 + \frac{\gamma_1(\hat\psi)}{\lbrace \rho j_\text{p}^{1/2}(\hat\psi) \rbrace} q   +  \frac{\gamma_2(\hat\psi)}{\rho^2 j_\text{p}(\hat\psi)} q^2   + O_p(n^{-3/2}) \right]\\
&= |j_{\lambda \lambda} (\hat\theta)| \left\lbrace 1 + \frac{E}{n^{1/2}}q + \frac{ F}{n}q^2  + O_p(n^{-3/2}) \right\rbrace.
\end{align*}
Therefore,
\begin{align}
\rho &=  \left\lbrace\frac{|j_{\lambda \lambda} (\hat\theta)|}{ |j_{\lambda \lambda} (\hat\theta_{\psi_0})|} \right\rbrace^{1/2} = \left[\frac{1}{ 1 + Cq/n^{1/2} + Dq^2/n + O_p(n^{-3/2}) } \right]^{1/2} \nonumber \\
 &= 1 + \frac{G}{n^{1/2}}q + \frac{ H}{n}q^2  + O_p(n^{-3/2}). \label{eq:rho_expansion}
\end{align}
Note that $\gamma_1(\psi)= O_p(1)$ and $\gamma_2(\psi)= O_p(1)$ by Lemma \ref{lemma:constrained_deriv} in \S7.2. Combining (\ref{eq:t_expansion}) and (\ref{eq:rho_expansion}), 
\begin{align*}
q &= t\rho = r\left\lbrace 1 + \frac{C}{n^{1/2}}q + \frac{D}{n}q^2 + O_p(n^{-3/2})\right\rbrace \left\lbrace 1 + \frac{G}{n^{1/2}}q + \frac{ H}{n}q^2  + O_p( n^{-3/2})   \right\rbrace.\\
r &= q \left\lbrace 1 + \frac{qC}{n^{1/2}} + \frac{q^2 D}{n}  + O_p(n^{-3/2})\right\rbrace^{-1} \left\lbrace 1 + \frac{G}{n^{1/2}}q + \frac{ H}{n}q^2  + O_p(n^{-3/2})   \right\rbrace^{-1} \\
&= q \left\lbrace 1 + \frac{A}{ n^{1/2} }q +\frac{B}{n}q^2 + O_p(n^{-3/2})\right\rbrace = q + \frac{A}{n^{1/2} }q^2 + \frac{ B}{n}q^3 + O_p(n^{-3/2}), 
\end{align*}
which shows the desired result.
For the location scale family, we use the same arguments and apply them to $s$ instead of $t$.
\end{proof}

\section{Simulations}\label{Sec6}
We perform numerical simulations to illustrate the results in \S3; for the linear exponential and location scale family
\begin{align}\label{eq:poly_r1}
r^\star   =  A^\star/n^{1/2} + (1+B^\star/n) r +  O_p(n^{-3/2}),
\end{align}
by Remark \ref{Remark:linearity}. To provide numerical evidence that (\ref{eq:poly_r1}) holds, note that, 
\begin{align}\label{eq:poly_r}
r^\star -  A^\star/n^{1/2} - (1+B^\star/n) r   =   O_p(n^{-3/2}),
\end{align}
and a sufficient condition for a random variable to be $O_p(n^{-3/2})$ is for both its mean and standard deviation to be $O(n^{-3/2})$. 
We illustrate the relationship in (\ref{eq:poly_r}) graphically by plotting the value of the mean and standard deviation of $r^\star - A^{\star} - (1 + B^{\star})r$ as a function of $n$; we expect the mean and standard deviation to roughly follow a linear trend when plotted against $\log(n)$, with slope $-3/2$ or smaller.

\subsection{Logistic Regression}
We first consider an example based on logistic regression in which there are $5$ covariates associated with each $y_i$, taken to be independent and identically distributed standard normals. The true regression coefficients are $\beta = (0, 1, 1,1,1)$, and the intercept is $\beta_0 = 1$. We are interested in testing for $H_0 : \beta_1 = 0$.
The number of samples is $n = 150, 300, 600, 1200, 2400$; we obtain estimates of $A^{\star}$ and $B^{\star}$ in  (\ref{eq:poly_r}) by numerical differentiation of the profile log-likelihood and the log-determinant of the information matrix. 
For each $n$, we simulate 2000 values of $r^\star - A^{\star} - (1 + B^{\star})r$ and plot the 95\% bootstrap confidence intervals of the empirical mean and standard deviation from 1000 bootstrap simulations. 
In Figure, \ref{fig:logistic_diff}, the mean and standard deviation of  $r^\star -A^{\star} - (1 + B^{\star})r$ are linear with slope $-3/2$, although the fit for the standard deviation is a bit off.

\begin{figure}[p]
    \centering
    \subfloat[Empirical mean of $r^\star -A^{\star} - (1+B^{\star}) r$.]{{\includegraphics[width=6cm]{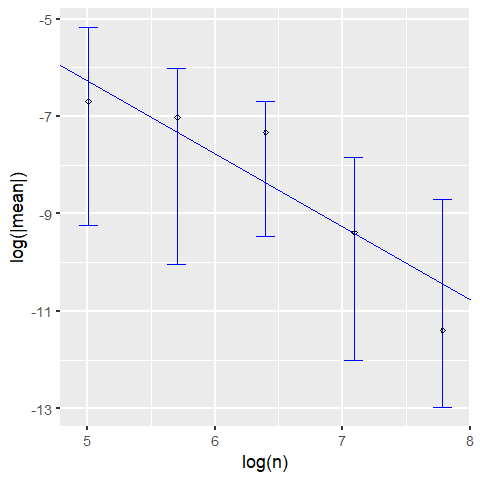} }}%
    \qquad
    \subfloat[Empirical standard deviation of $r^\star -A^{\star} - (1+B^{\star}) r$]{{\includegraphics[width=6cm]{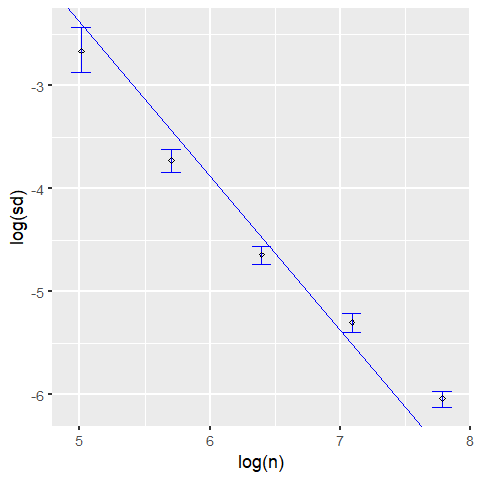} }}%
    \caption{Plots for logistic regression illustrating order of $r^\star -A^{\star} - (1+B^{\star}) r$. We see that both the mean and standard deviation roughly follows a slop of $-3/2$ as a function of $n$ on the log scale, suggesting that the difference is of order $O_p(n^{-3/2})$. The solid line has slope $-3/2$ and the intercept is fitted.}
    \label{fig:logistic_diff}%
\end{figure}

\subsection{Linear Regression with $t$ Errors}
We consider an example based on a location-scale regression model where the error follows a $t_5$-distribution. 
There are $5$ covariates $x_i$, which are taken to be independent and identically distributed standard normal. 
We are interested in testing for $H_0: \beta_1 = 0$ 
The true regression coefficients are $\beta = (0, 1, 1, 1,1)$, and the intercept is $\beta_0 = 1$. 
The number of samples is $n =150, 300, 600, 1200, 2400$. 
For each $n$, we simulate 2000 values of $r^\star - A^{\star} - (1 + B^{\star})r$ and plot the 95\% bootstrap confidence intervals of the empirical mean and standard deviation from 1000 bootstrap simulations. 
The results of the simulations are given in Figure \ref{fig:t_diff}.

\begin{figure}[p]
    \centering
    \subfloat[Empirical mean of $r^\star -A^{\star} - (1+B^{\star}) r$ .]{{\includegraphics[width=6cm]{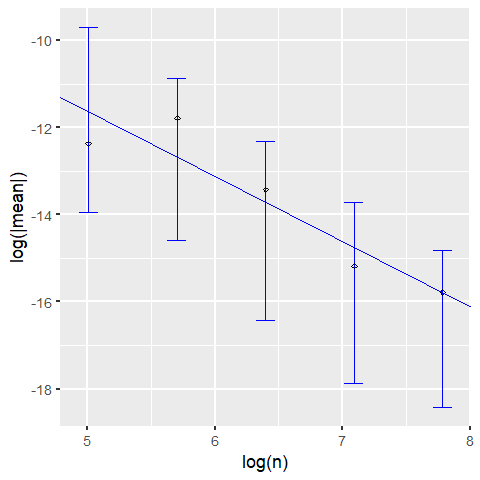} }}%
    \qquad
    \subfloat[Empirical standard deviation of $r^\star -A^{\star} - (1+B^{\star}) r$]{{\includegraphics[width=6cm]{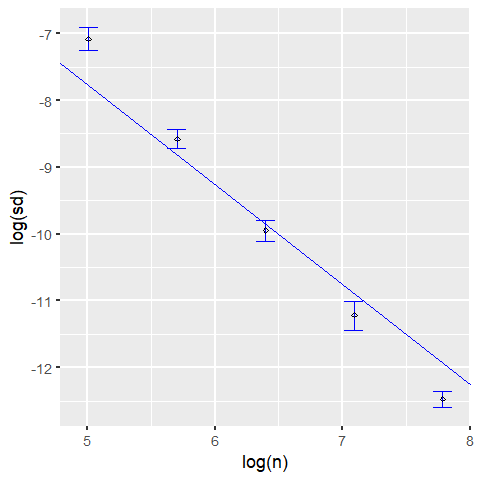} }}%
    \caption{Plots for $t_5$ based regression illustrating the order of $r^\star - A^{\star} - (1 + B^{\star})r$. We see that both the mean and standard deviation roughly follows a slope of $-3/2$ as a function of $n$ on the log scale, suggesting that the difference is of order $O_p(n^{-3/2})$. The solid line has slope $-3/2$ and the intercept is fitted.} 
    \label{fig:t_diff}%
\end{figure}

\section{Additional Proof Details}\label{Appendix}
\subsection{Order of $\gamma_k$ and $\kappa_k$}
We first establish the order of $\gamma_k(\hat\psi)$ and $\zeta_k(\hat\psi)$ for arbitrary integer $k$. 
All results are proved for $\psi$ and $\lambda$ scalar parameters; generalization to vector $\lambda$ is straightforward but notationally tedious.

\begin{lemma}\label{lemma:constrained_deriv}
For all integers $k$,

\begin{gather}
\frac{\partial^k\hat\lambda_\psi}{\partial\psi^k}|_{\psi = \hat\psi} = O_p(1), \quad \frac{d^k}{d\psi^k} \tilde\jmath_{\lambda\psi}|_{\psi = \hat\psi } = O_p(n), \quad \frac{d^k}{d\psi^k} \tilde \jmath_{\lambda\lambda}|_{\psi = \hat\psi} = O_p(n),\nonumber \\
 \kappa_k(\hat\psi) = O_p\lbrace n^{ -(k-2)/2}\rbrace, \quad \gamma_k(\hat\psi) = O_p(1). \nonumber
\end{gather}

\end{lemma}

\begin{proof}
Differentiating the expression
\begin{align*}
    0 &= \tilde{l}_{\lambda},
\end{align*}
we obtain, 
\begin{align}
    \tilde{\jmath}_{\lambda\lambda}\frac{\partial\hat\lambda_{\psi}}{\partial\psi} &=  -\tilde{\jmath}_{\psi \lambda} \label{eqn:2},
\end{align}    
thus, 
\begin{align}
    \frac{\partial\hat\lambda_{\psi}}{\partial\psi}|_{\psi = \hat\psi} &=  -\hat{\jmath}^{-1}_{\lambda\lambda} \hat{\jmath}_{\psi \lambda} = O_p(1) \label{eqn:3},
\end{align}
by Assumptions \ref{ass:smooth} and \ref{ass:eigen}, we use the above to show  
\[\frac{d}{d\psi} \tilde{\jmath}_{\lambda\lambda}|_{\psi = \hat\psi}  = \hat\jmath_{\psi\psi} + \frac{\partial\hat\lambda}{\partial\psi}|_{\psi = \hat\psi}  \ \hat\jmath_{\psi\lambda} = O_p(n), \] 
and by the same argument show that $d\tilde\jmath_{\lambda\psi}/d\psi|_{\psi = \hat{\psi}} = O_p(n)$. 
This proves the first three claims hold for $k = 1$.

\textbf{Induction Step:} 
We assume that for all $k_1 \leq k$ the result holds, we need to show that for $k + 1$

\[\frac{\partial^{k+1}\hat\lambda_\psi}{\partial\psi^{k+1}}|_{\psi = \hat\psi} = O_p(1),\quad \frac{d^{k+1}}{d\psi^{k+1}} \tilde{\jmath}_{\lambda\lambda}|_{\psi = \hat\psi} = O_p(n), \quad \frac{d^{k+1}}{d\psi^{k+1}} \tilde{\jmath}_{\lambda\psi}|_{\psi = \hat\psi} = O_p(n). \]

First, we differentiate (\ref{eqn:3}) $k$ times to obtain 
\[ \frac{\partial^{k+1}\hat\lambda_{\psi}}{\partial\psi^{k+1}} = - \sum_{i = 1}^{k} {k \choose i} \frac{d^i  (\tilde{\jmath}_{\lambda\lambda})^{-1}}{d\psi^{i}} \frac{d^{k - i }  \tilde{\jmath}_{\psi \lambda} }{d\psi^{k - i }}.   \]

Using Faa di Bruno's formula for the differentiation of a composition of functions we obtain 
\begin{align}
  \frac{d^i}{d\psi^i} (\tilde{\jmath}_{\lambda\lambda})^{-1} = \sum_{k = 1}^{i} (-1)^{k} k! \lbrace	\tilde{\jmath}_{\lambda\lambda}\rbrace^{(-k-1)} B_{i,k}\left( \frac{d}{d\psi} \tilde\jmath_{\lambda\lambda}, \cdots , \frac{d}{d\psi^{i - k +1}} \tilde\jmath_{\lambda\lambda}  \right), \label{eqn:bruno_di_faa} 
 \end{align}
where
 \[B_{i,k}(x_1, \cdots, x_{i - k +1}) = \sum \frac{i!}{j_1! j_2! \cdots j_{i-k+1}!}  \left( \frac{x_1}{1!} \right)^{j_1} \left(\frac{x_2}{2!}\right)^{j_2} \cdots \left(\frac{x_{i - k +1}}{(i - k +1)!}\right)^{j_{i - k +1}}, \]
and the summation in the above expression is taken over all sets of ${j_1, \dots, j_{i - k +1}}$ such that,
\begin{align*}
j_1 +j_2 +\cdot +j_{i - k+1} = k, \quad j_1 +2j_2 + \cdots + (i- k +1)j_{i - k +1} = i.
\end{align*}
The polynomials $B_{i,k}$ are the partial Bell polynomials. 
From (\ref{eqn:bruno_di_faa}) we deduce that $d (\tilde{\jmath}_{
\lambda\lambda})^{-1}/d\psi|_{\psi = \hat\psi} = O_p(n^{-1})$, since the constraint $j_1 +j_2 +\cdot +j_{n - l+1} = k$ implies that
\[B_{i,k}(d \tilde{\jmath}_{\lambda\lambda}/d\psi|_{\psi = \hat\psi}, \cdots, d^{i - k+1} \tilde{\jmath}_{\lambda\lambda}/d\psi^{i - k+1}|_{\psi = \hat\psi}) = O_p(n^k),
\]
 and $( \hat\jmath_{\lambda\lambda} )^{-k - 1} = O_p(n^{-k-1}) $, which implies that every term in the summation is $O_p(n^{-1})$.

Thus, by the induction assumption, we have $d^{k - i }  \tilde{\jmath}_{\psi \lambda} /d\psi^{k - i }|_{\psi = \hat\psi}  = O_p(n)$ for $i = 1, \cdots, k - 1$. 
Therefore we have the desired result for the constrained derivative of the maximum likelihood estimate. 
Next we show that 
\[ \frac{d^{k+1}}{d\psi^{k+1}} \tilde{\jmath}_{\lambda\lambda}|_{\psi = \hat\psi} = O_p(n), \quad \frac{d^{k+1}}{d\psi^{k+1}} \tilde{\jmath}_{\lambda\psi}|_{\psi = \hat\psi} = O_p(n). \]

For this,
\begin{align}
\frac{d^{k+1}}{d\psi^{k+1}} \tilde{\jmath}_{\lambda\psi}  = \sum_{i,j,l,m = 1}^{k+1} a_{i,j,l,m} \frac{\partial^{i + j} \tilde{\jmath}_{\psi\lambda} (\psi , \hat\lambda_\psi)}{\partial\psi^{i}\partial\lambda^{j}} \left( \frac{\partial^l\hat\lambda_\psi}{\partial\psi^l}  \right)^m  = O_p(n), \label{eq:chain_rule}
\end{align}   
which can be obtained through successive applications of the chain rule, some of the coefficients $a_{i,j,l,m}$ may be $0$. The result follows from the fact that all derivatives of the constrained maximum likelihood estimate are $O_p(1)$ up to the $(k+1)$ order when evaluated at $\hat\psi$ and log-likelihood derivatives are assumed to be $O_p(n)$ when evaluated at $\hat\theta$. 
A similar argument can be made for the derivatives of $\tilde{\jmath}_{\lambda\lambda}$. 

\textbf{Order of $\kappa_k(\psi)$}

The total derivative of the profile log-likelihood function is a summation of partial derivatives multiplied by the derivative of the constrained maximum likelihood estimate, so the result is obtained from arguments used in (\ref{eq:chain_rule}).

 

\textbf{Order of $\gamma_k(\psi)$}

We have 
\begin{align}
\gamma_k(\psi) = \sum_{i + j = k, \ j \geq 1} \tr\left[ \frac{d^i}{d\psi^i} (\tilde{\jmath}_{\lambda\lambda})^{-1} \frac{d^j}{d\psi^j} \tilde{\jmath}_{\lambda\lambda} \right]. \label{eq:gamma_k}
\end{align}  
 Using $d^i (\tilde{\jmath}_{\lambda\lambda})^{-1}/d\psi^i|_{\psi = \hat\psi} = O_p(n^{-i})$ from (\ref{eqn:bruno_di_faa}) and  $d^j \tilde\jmath_{\lambda\lambda} /d\psi^j|_{\psi = \hat\psi} = O_p(n)$ from (\ref{eq:chain_rule}) we conclude $\gamma_k(\hat\psi) = O_p(1)$.
\end{proof}
%
%
%

\subsection{Order of Maximum Singular Value in \S\ref{Sec4}}\label{appendix:max}
We obtain the order of the maximum singular value of the second derivative of the information matrix for the location-scale model in the high-dimensional setting. We have
\begin{align*}
\frac{d^2}{d\psi^2} \tilde\jmath_{\lambda\lambda}|_{\psi = \hat\psi} &=  \hat\jmath_{\psi\psi\lambda\lambda} + 2 \sum_{i = 1}^{p-1} \frac{\partial\hat\lambda_{\psi, i }}{\partial\psi} |_{\psi = \hat\psi} \ \hat\jmath_{\psi\lambda_i\lambda\lambda} +  \sum_{i = 1}^{p-1} \frac{\partial^2\hat\lambda_{\psi, i }}{\partial\psi^2} |_{\psi = \hat\psi} \ \hat\jmath_{\lambda_i\lambda\lambda}\\
&+ \sum_{i = 1}^{p-1} \sum_{j = 1}^{p-1} \frac{\partial\hat\lambda_{\psi, i }}{\partial\psi}|_{\psi = \hat\psi} \frac{\partial\hat\lambda_{\psi, j }}{\partial\psi} |_{\psi = \hat\psi} \ \hat\jmath_{\lambda_i\lambda_j\lambda\lambda}.
\end{align*} 
Now the maximal singular values of the matrices of interest are:
\begin{align*}
&\sigma_{\max} \lbrace \hat\jmath_{\psi\psi\lambda\lambda} \rbrace \leq \norm{ \hat\jmath_{\psi\psi\lambda\lambda}}_F = O_p(pn). \\
&\sigma_{\max} \left\lbrace\sum_{i = 1}^{p-1} \frac{\partial\hat\lambda_{\psi, i }}{\partial\psi} |_{\psi = \hat\psi} \ \hat\jmath_{\psi\lambda_i\lambda\lambda} \right\rbrace \\
&\leq \sum_{i = 1}^{p-1} \left| \frac{\partial\hat\lambda_{\psi, i }}{\partial\psi} |_{\psi = \hat\psi} \right| \sigma_{\max} \left\lbrace \hat\jmath_{\psi\lambda_i\lambda\lambda}\right\rbrace \\
&\leq \sum_{i = 1}^{p-1} \left| \frac{\partial\hat\lambda_{\psi, i }}{\partial\psi} |_{\psi = \hat\psi} \right| \norm{  \hat\jmath_{\psi\lambda_i\lambda\lambda}}_F \\
&\leq p^{1/2} \norm{\frac{\partial\hat\lambda_{\psi }}{\partial\psi} |_{\psi = \hat\psi}}_2 \max_{i= 1, \dots, p} \norm{  \hat\jmath_{\psi\lambda_i\lambda\lambda}}_F  = O_p(p^{2} n^{1/2}).\\
&\sigma_{\max} \left\lbrace \sum_{i = 1}^{p-1} \frac{\partial^2\hat\lambda_{\psi, i }}{\partial\psi^2} |_{\psi = \hat\psi} \ \hat\jmath_{\lambda_1\lambda\lambda} \right\rbrace\\
&\leq  \sum_{i = 1}^{p-1} \left|\frac{\partial^2\hat\lambda_{\psi, i }}{\partial\psi^2} |_{\psi = \hat\psi} \right| \ \sigma_{\max} \left\lbrace \hat\jmath_{\lambda_i\lambda\lambda} \right\rbrace = O_p(pn) \\
&\sigma_{\max} \left\lbrace \sum_{i = 1}^{p-1} \sum_{j = 1}^{p-1} \frac{\partial\hat\lambda_{\psi, i }}{\partial\psi}|_{\psi = \hat\psi} \frac{\partial\hat\lambda_{\psi, j }}{\partial\psi} |_{\psi = \hat\psi} \ \hat\jmath_{\lambda_i\lambda_j\lambda\lambda} \right\rbrace \\
&\leq  \sum_{i = 1}^{p-1} \sum_{j = 1}^{p-1} \left| \frac{\partial\hat\lambda_{\psi, i }}{\partial\psi}|_{\psi = \hat\psi}\right| \left| \frac{\partial\hat\lambda_{\psi, j }}{\partial\psi} |_{\psi = \hat\psi} \right| \ \sigma_{\max} \left\lbrace \hat\jmath_{\lambda_i\lambda_j\lambda\lambda} \right\rbrace\\
&\leq  \sum_{i = 1}^{p-1} \sum_{j = 1}^{p-1} \left|\frac{\partial\hat\lambda_{\psi, i }}{\partial\psi}|_{\psi = \hat\psi}\right| \left| \frac{\partial\hat\lambda_{\psi, j }}{\partial\psi} |_{\psi = \hat\psi} \right| \norm{ \hat\jmath_{\lambda_i\lambda_j\lambda\lambda}}_F \\
&= p \norm{\frac{\partial\hat\lambda_{\psi }}{\partial\psi}|_{\psi = \hat\psi}}_2^2 \max_{i,j = 1, \dots, p} \norm{ \hat\jmath_{\lambda_i\lambda_j\lambda\lambda}}_F = O_p(p^3).
\end{align*}
Using the triangle inequality results in the rates obtained for $\gamma_1(\hat\psi)$ and $\gamma_2(\hat\psi)$ given in \S 5.2.
\section*{Acknowledgements}

We thank Nicola Sartori, Michele Lambardi di San Miniato, Ioannis Kosmidis, Heather Battey and Micha\"el Lalancette for helpful discussions. This research was partially supported by the Natural Sciences and Engineering Research Council of Canada and the Vector Institute.
\par


\bibhang=1.7pc
\bibsep=2pt
\fontsize{9}{14pt plus.8pt minus .6pt}\selectfont
\renewcommand\bibname{\large \bf References}
\expandafter\ifx\csname
natexlab\endcsname\relax\def\natexlab#1{#1}\fi
\expandafter\ifx\csname url\endcsname\relax
  \def\url#1{\texttt{#1}}\fi
\expandafter\ifx\csname urlprefix\endcsname\relax\def\urlprefix{URL}\fi

\bibliographystyle{chicago}      
\bibliography{biblio}   

\end{document}